\documentclass[a4paper]{article}
\usepackage{color}
\usepackage{amsmath}
\usepackage{amsthm}
\usepackage{amssymb}
\usepackage{amsmath}
\usepackage{mathcomp}
\usepackage{dsfont}

\newtheorem{definition}{Definition}[section]
\newtheorem{lemma}[definition]{Lemma}
\newtheorem{proposition}[definition]{Proposition}
\newtheorem{theorem}[definition]{Theorem}

\numberwithin{equation}{section}

\def\eps{\varepsilon}
\def\rd{\mathrm{d}}
\def\ri{\mathrm{i}}
\def\UN{\mathcal{U}_N\left(t;0\right)}
\def\tr{\mathrm{tr}}
\def\Re{\mathrm{Re}}
\def\bR{\mathbb{R}}
\def\bN{\mathbb{N}}

\def\cF{\mathcal{F}}
\def\cE{\mathcal{E}}
\def\cH{\mathcal{H}}
\def\cN{\mathcal{N}}
\def\cU{\mathcal{U}}
\def\ph{\varphi}

\title{Mean field evolution of fermions \\ with Coulomb interaction}
\author{Marcello Porta, Simone Rademacher, Chiara Saffirio   
\\ and Benjamin Schlein \\
\\
Institute of Mathematics, University of Zurich\\
Winterthurerstrasse 190, 8057 Zurich, Switzerland}

\begin{document}

\maketitle

\begin{abstract}
We study the many body Schr\"odinger evolution of weakly coupled fermions interacting through a Coulomb potential. 
We are interested in a joint mean field and semiclassical scaling, that emerges naturally for initially confined particles. For initial data describing approximate Slater determinants, we prove convergence of the many-body evolution towards Hartree-Fock dynamics. Our result holds under a condition on the solution of the Hartree-Fock equation, that we can only show in a very special situation (translation invariant data, whose Hartree-Fock evolution is trivial), but that we expect to hold more generally.
\end{abstract}

\section{Introduction}

The evolution of a system of $N$ fermions in the mean field regime is described by the Schr\"odinger equation
\begin{equation}\label{eq:schr0} i \eps \partial_t \psi_{N,t} = \left[ \sum_{j=1}^N \left[ -\eps^2 \Delta_{x_j} + V_\text{ext} (x_j) \right] + \frac{1}{N} \sum_{i<j}^N V(x_i - x_j) \right] \psi_{N,t} \, ,
\end{equation}
in the limit $N \to \infty$. Here $\eps = N^{-1/3}$ and, according to fermionic statistics, $\psi_{N,t} \in L^2_a (\bR^{3N})$, the subspace of $L^2 (\bR^{3N})$ consisting of wave functions antisymmetric with respect to permutations of the $N$ particles. 


The Schr\"odinger equation (\ref{eq:schr0}) is relevant for initial $N$ particle wave functions $\psi_{N,0} \in L^2_a (\bR^{3N})$ localized in a volume of order one; in this case, the factor $\eps^2$ in front of the kinetic energy guarantees that both terms in the Hamiltonian are of order 
$N$. We conclude that, for fermionic systems, the mean field regime is linked with a semiclassical limit, with $\eps = N^{-1/3}$ playing the role of Planck's constant (notice, however, that in other situations, different scalings may be of interest; see, in particular, \cite{BGGM,FK,PP,BBPPT}).

Physically, it makes sense to consider initial data approximating equilibria of confined systems. At zero temperature, this leads to the study of the 
mean field dynamics of approximate Slater determinants. In \cite{BPS}, it has been proven that this evolution can be described through the Hartree-Fock equation, for regular interaction (the same conclusion was already reached in \cite{EESY}, for analytic potentials and short times). At positive temperature, convergence towards Hartree-Fock dynamics for mixed quasi free initial data has been later established in \cite{BJPSS}. 

Let us focus on the zero temperature case and explain the results of \cite{BPS} in more details. Let $\omega_N$ be a sequence of orthogonal projections on $L^2 (\bR^3)$ with $\tr \, \omega_N = N$ and such that 
\begin{equation}\label{eq:comm} \tr \, |[\omega_N , x]| \leq C N \eps , \qquad \text{and } \quad \tr \, |[\omega_N, \eps \nabla ] | \leq C N \eps \, . 
\end{equation}
The projections $\omega_N$ are the one-particle reduced densities of $N$-particle Slater determinants. We consider 
the time evolution of initial  fermionic wave functions $\psi_N \in L^2_a (\bR^{3N})$ with one-particle reduced density $\gamma^{(1)}_N$ close in the trace norm topology to $\omega_N$.  Denoting by $\psi_{N,t}$ the solution of the Schr\"odinger equation (\ref{eq:schr0}) with initial data $\psi_N$ and by $\gamma_{N,t}^{(1)}$ the corresponding one-particle reduced density, it is shown in \cite{BPS} that $\gamma_{N,t}^{(1)}$ is close (in the Hilbert-Schmidt and in the trace class topology) to the solution of the Hartree-Fock equation 
\begin{equation}\label{eq:HF0} i\eps \partial_t \omega_{N,t} = \left[ -\eps^2 \Delta + (V* \rho_t) - X_t , \omega_{N,t} \right] 
\end{equation}
with initial data $\omega_{N,0} = \omega_N$. Here $\rho_t (x) = N^{-1} \omega_{N,t} (x;x)$ and the exchange operator $X_t$ is defined by the integral kernel $X_t (x;y) = N^{-1} V (x-y) \omega_{N,t} (x;y)$ (strictly speaking, in \cite{BPS} the convergence towards the Hartree-Fock equation has only been shown for $V_\text{ext} (x) = 0$, but it is easy to extend the result to non-vanishing smooth external fields). 

In other words, the time-evolution of initial data close to a Slater determinant remains close to a Slater determinant evolved with respect to the Hartree-Fock equation (\ref{eq:HF0}). This holds provided the reduced density 
$\omega_N$ of the initial Slater determinant satisfies the commutator bounds (\ref{eq:comm}). These estimates play a crucial role in \cite{BPS} to obtain convergence up to the correct time scale. They reflect the semiclassical  structure of $\omega_N$, i.e. the fact that the integral kernel $\omega_N (x;y)$ varies on the short scale of order $\eps$ in the $x-y$ direction, while it varies on scales of order one in the $x+y$ direction. This structure is expected to arise in Slater determinants approximating equilibrium states. 

Notice that the Hartree-Fock equation (\ref{eq:HF0}) still depends on $N$ (recall that $\eps = N^{-1/3}$). As $N \to \infty$, one expects the Wigner transform 
\[ W_{N,t} (x,v) = \int dy \, \omega_{N,t} \left( x + \frac{\eps y}{2} ;x- \frac{\eps y}{2} \right) e^{i v \cdot y} \]
of the solution of (\ref{eq:HF0}) to converge towards a probability density $W_{\infty,t}$ on phase space, solving the classical Vlasov equation
\[ \partial_t W_{\infty,t} + v \cdot \nabla_x W_{\infty,t} + \nabla (V * \rho_t) \cdot \nabla_v W_{\infty,t} = 0 \]
Convergence of the Hartree-Fock evolution towards the Vlasov dynamics has been established in several works, see  \cite{LP,MM,APPP,AKN}, but only recently, in \cite{BPSS}, some results have been obtained for the situation we consider here, where $\omega_{N,t}$ is a projection. Remark also that direct convergence from the many-body quantum evolution to the Vlasov dynamics has been shown in \cite{NS} for analytic potentials and later in \cite{Sp} for $C^2$-interactions.  

The convergence towards the Hartree-Fock equation has been established in \cite{BPS} for regular interaction potentials satisfying 
\begin{equation}\label{eq:whV} \int dp\, |\widehat{V} (p)| (1+p^2)   < \infty\,. 
\end{equation}
This assumption excludes the case of a Coulomb interaction $V(x) = 1/|x|$. The Schr\"odinger equation (\ref{eq:schr0}) for a Coulomb potential is very interesting from the point of view of physics. It arises naturally when considering the dynamics of large atoms and molecules. In fact the Hamilton operator for an electrically neutral atom with $N$ electrons and a nucleus fixed at the origin is given by 
\begin{equation}\label{eq:atom-ham0} H^\text{atom}_N = \sum_{j=1}^N \left[ -\Delta_{x_j} - \frac{N}{|x_j|} \right] + \sum_{i<j}^N \frac{1}{|x_i -x_j|} 
\end{equation}
and acts on the Hilbert space $L^2_a (\bR^{3N})$ of the $N$ electrons. Thomas-Fermi theory suggests that electrons are localized at distances of order $N^{-1/3}$ from the nucleus (see, for example, the review article \cite{L}). It is therefore convenient to introduce new variables $X_j = N^{1/3} x_j$. Expressed in terms of the new variables, the atomic Hamiltonian (\ref{eq:atom-ham0}) takes the form 
\begin{equation}\label{eq:Hatom-N} \begin{split} H^\text{atom}_N &= \sum_{j=1}^N \left[ -N^{2/3} \Delta_{X_j}  - \frac{N^{4/3}}{|X_j|} \right] + N^{1/3} \sum_{i<j}^N \frac{1}{|X_i - X_j|} \\ &= N^{4/3} \left\{ \sum_{j=1}^N \left[ -\eps^2 \Delta_{X_j} - \frac{1}{|X_j|} \right] + \frac{1}{N} \sum_{i<j} \frac{1}{|X_i -X_j|} \right\} \end{split} \end{equation}
with $\eps = N^{-1/3}$. Choosing the correct time scale, we arrive exactly at the Schr\"odinger equation (\ref{eq:schr0}) with $V_\text{ext} (x) = - 1/|x|$ and interaction $V(x) = 1/|x|$. Remark that Hartree-Fock theory is known to provide a good approximation to the ground state energy of (\ref{eq:Hatom-N}). While the classical Thomas-Fermi theory only captures the leading order of the ground state energy, which is of order $N^{7/3}$ (see \cite{LSi,L}), Hartree-Fock theory was proven in \cite{B,GS} to provide a much more accurate approximation, with an error of order smaller than $N^{5/3}$.

The goal of our paper is to extend the convergence of the many-body dynamics towards the time-dependent Hartree-Fock equation to the case of a Coulomb interaction. Our results are still not completely satisfactory, in the sense that they make use of a property of the solution of the time-dependent Hartree-Fock equation (\ref{eq:HF0}) which we can only show to hold true for very special choices of the initial data. 
Nevertheless, we believe our results to be of some interest, since they reduce the problem of the derivation of the Hartree-Fock equation for Coulomb systems from the analysis of the many-body Schr\"odinger equation (\ref{eq:schr0}) to the study of the properties of the simpler Hartree-Fock equation (\ref{eq:HF0}). Notice that the time evolution of fermions interacting through a Coulomb potential has been recently considered in \cite{BBPPT}. In this work, however, a different scaling was considered, with the $N$ particles occupying a large volume of order $N$. After rescaling lengths, this choice leads to the Schr\"odinger equation (\ref{eq:schr0}), with short times $t$ of order $\eps = N^{-1/3}$. 

Let us now illustrate our results in a precise form. For a wave function $\psi_N \in L^2_a (\bR^{3N})$ we define the one-particle reduced density $\gamma^{(1)}_N$ as the non-negative trace class operator with integral kernel given by  
\begin{equation}\label{eq:gamma1}\gamma^{(1)}_N (x;y) = N \int dx_2 \dots dx_N\,\psi_N (x, x_2, \dots , x_N) \overline{\psi}_N (y, x_2, \dots , x_N)\,.
\end{equation}
Notice here that we use the standard normalization $\tr \, \gamma^{(1)}_N = N$. A simple computation shows that the reduced density of the Slater determinant 
\[ \psi_\text{slater} (x_1, \dots , x_N) = \frac{1}{\sqrt{N!}} \det \left( f_i (x_j) \right)_{i,j \leq N} \, , \]
where $\{ f_j \}_{j=1}^N$ is an orthonormal system on $L^2 (\bR^3)$, is given by the orthogonal projection 
\begin{equation}\label{eq:omega0} \omega_N = \sum_{j=1}^N |f_j \rangle \langle f_j| \end{equation}
on the $N$ dimensional linear space spanned by the orbitals $\{ f_j \}_{j=1}^N$. 

We consider a sequence of initial data $\psi_N \in L^2_a (\bR^{3N})$, which we assume close to a Slater determinant in the sense that the one-particle reduced density $\gamma^{(1)}_N$  associated with $\psi_N$ satisfies $\| \gamma^{(1)}_N - \omega_N \|_\text{tr} \leq C$, uniformly in $N$, for a sequence $\omega_N$ of orthogonal projections of rank $N$ ($\omega_N$ is the one-particle reduced density of a Slater determinant). 

Under this condition, we consider the evolution 
$\psi_{N,t} = e^{-i H_N t/\eps} \psi_N$ of the initial data $\psi_N$, generated by the Coulombic Hamiltonian 
\begin{equation}\label{eq:ham-mf} H_N = \sum_{j=1}^N -\eps^2 \Delta_{x_j} + \frac{1}{N} \sum_{i<j}^N \frac{1}{|x_i - x_j|}\,.
\end{equation} 
To simplify the notation we assumed here that the external potential vanishes (but it is easy to extend our results to the case $V_\text{ext} \not = 0$). 

We compare $\psi_{N,t}$ with the Slater determinant with reduced density $\omega_{N,t}$ given by the solution of the time-dependent Hartree-Fock equation 
\begin{equation}\label{eq:HF} i\eps \partial_t \omega_{N,t} = \left[ -\eps^2 \Delta + \frac{1}{|.|} * \rho_t -X_t , \omega_{N,t} \right] 
\end{equation}
with the position-space density $\rho_t (x) = N^{-1} \omega_{N,t} (x;x)$ and where $X_t$ is the exchange operator, with the integral kernel $X_t (x;y) = N^{-1} |x-y|^{-1}$. 

As in \cite{BPS}, a crucial role in our analysis is played by the operator $|[x,\omega_{N,t}]|$. Let us define its density 
\[ \rho_{|[x,\omega_{N,t}]|} (x) = |[x,\omega_{N,t}]| (x;x)\,. \]
An important ingredient in \cite{BPS} was the estimate 
\begin{equation}\label{eq:L1bd} \| \rho_{|[x,\omega_{N,t}]|} \|_1 = \tr \, |[ x,\omega_{N,t}]| \leq C e^{K|t|} N \eps  \, , 
\end{equation}
valid for all $t \in \bR$. For interaction potentials satisfying (\ref{eq:whV}), (\ref{eq:L1bd}) was proven in \cite{BPS} propagating the commutator bounds (\ref{eq:comm}) along the solution of the Hartree-Fock equation. Here, to deal with the Coulomb singularity of the interaction, we need additional information on the operator $|[x,\omega_{N,t}]|$; in particular, we need a bound (again of the order $N\eps$) on the $L^p$ norm of $\rho_{|[x,\omega_{N,t}]|}$, for a $p > 5$. Unfortunately, we do not know what assumptions on the initial data $\omega_N$ imply the validity of these bounds for the solution of the Hartree-Fock equation (\ref{eq:L1bd}). Our main result is therefore a conditional statement; it gives convergence of the many-body evolution with Coulomb interaction towards the Hartree-Fock equation on the time interval $[0;T]$ provided the $L^1$ and the $L^p$ norm of $\rho_{|[x,\omega_{N,t}]|}$ are of order $N\eps$, uniformly in $t \in [0;T]$ (for a $p > 5$). 
\begin{theorem}\label{thm:main}
Let $\omega_N$ be a sequence of orthogonal projections on $L^2 (\bR^3)$, with $\tr \, \omega_N = N$ and such that $\tr\, (-\eps^2 \Delta) \, \omega_N \leq C N$, for a constant $C >0$ independent of $N$. Let $\omega_{N,t}$ denote the solution of the Hartree-Fock equation (\ref{eq:HF}) with initial data $\omega_{N,0} = \omega_N$.  We assume that there exists a time $T > 0$, a $p > 5$ and a constant $C > 0$ such that 
\begin{equation}\label{eq:ass-main} \sup_{t \in [0;T]}  \, \sum_{i=1}^3 \left[ \| \rho_{|[x_i,\omega_{N,t}]|} \|_1 + \|  \rho_{|[x_i,\omega_{N,t}]|} \|_p \right] \leq C N \eps \,.  \end{equation}

Let $\psi_N \in L^2_a (\bR^{3N})$ be such that its one-particle reduced density matrix $\gamma_{N}^{(1)}$ satisfies
\begin{equation}\label{eq:conden} \tr \, \left| \gamma^{(1)}_N - \omega_N \right| \leq C N^\alpha \end{equation}
for a constant $C > 0$ and an exponent $0 \leq \alpha < 1$.

Consider the evolution $\psi_{N,t}= e^{-iH_N t/\eps} \psi_N$, with the Hamilton operator (\ref{eq:ham-mf}) and let  
$\gamma^{(1)}_{N,t}$ be the corresponding one-particle reduced density. Then for every $\delta > 0$ there exists $C >0$ such that 
\begin{equation}\label{eq:HS-bd} 
\sup_{t \in [0;T]} \, \left\| \gamma_{N,t}^{(1)} - \omega_{N,t} \right\|_\text{HS} \leq C \left[ N^{\alpha/2}+ N^{5/12 + \delta} \right] 
\end{equation}
and 
\begin{equation}\label{eq:tr-bd} 
\sup_{t \in [0;T]} \, \tr \left| \gamma_{N,t}^{(1)} - \omega_{N,t} \right| \leq C \left[ N^\alpha + N^{11/12 + \delta}\right] . 
\end{equation}  
\end{theorem}

Recall that $\| \omega_{N,t} \|_\text{HS} = N^{1/2}$ and $\tr\, \omega_{N,t} = N$; this implies that the bounds (\ref{eq:HS-bd}) and (\ref{eq:tr-bd}) are non-trivial. They really show that the Hartree-Fock equation is a good approximation for the many-body evolution with a Coulomb interaction. Remark here that the exponent $0 \leq \alpha < 1$ measures the number of particles that, at time $t=0$, are not in the Slater determinant (the initial number of excitations). 

As pointed out in the introduction, the Hartree-Fock equation (\ref{eq:HF}) still depends on $N$. As $N \to \infty$, the Wigner transform of the solution of (\ref{eq:HF}) is expected to converge to a solution of the Vlasov equation. However, this result is still open. In fact, the result of \cite{LP}, which applies to the case of a Coulomb interaction, does not allow $\omega_{N,t}$ to be a projection. 

Despite the fact that we do not know how to prove the bounds (\ref{eq:ass-main}) for the solution of the Hartree-Fock equation, they are consistent with the idea that $\omega_{N,t}$ varies on a length scale of order $\eps$ in the $(x-y)$ direction, while it is regular and it varies on scales of order one in the $(x+y)$ direction. 

There is in fact one special situation, in which the required bounds can be easily shown to hold true. Consider namely an $N$-fermion system described on a finite box $\Lambda$ with volume of order one and periodic boundary conditions. In this case, we can consider translation invariant Slater determinants, whose reduced densities have integral kernels $\omega_N (x;y)$ depending only on $(x-y)$. Also the commutator $[x,\omega_N]$ and its absolute value $|[x,\omega_N]|$ are then translation invariant, and therefore $\rho_{|[x,\omega_N]|}$ is a constant, which we can reasonably assume to be of order $N\eps$ (meaning that $\omega_N$ is a function of $(x-y)$ decaying at distances $|x-y| \gg \eps$ from the diagonal). Then, we trivially have $\| \rho_{|[x,\omega_N]|} \|_p \leq C N\eps$ for all $1 \leq p \leq \infty$. Furthermore, it is easy to check that the Hartree-Fock evolution does not change translation invariant initial data, i.e. in this case we have $\omega_{N,t} = \omega_N$ for all $t \in \bR$. This means that $\| \rho_{|[x,\omega_{N,t}]|} \|_p \leq C N\eps$ for all $1\leq p \leq \infty$ and also for all $t \in \bR$. So, for translation invariant Slater determinants describing $N$ fermions in a box with volume of order one with periodic boundary conditions, Theorem \ref{thm:main} shows (in this case, with no further assumptions), that the many body evolution generated by (\ref{eq:ham-mf}) can be approximated by the Hartree-Fock equation, which means, in other words, that it leaves the state of the system approximately invariant (we stated Theorem \ref{thm:main} for systems defined on $\bR^3$, but the result and its proof can be easily extended to systems defined on a box with volume of order one and with periodic boundary conditions).  

Let us remark that Theorem \ref{thm:main} can be extended by including an external potential in the Hamilton operator (\ref{eq:ham-mf}) (and in the Hartree-Fock equation (\ref{eq:HF})). Of course, in presence of an external potential it may be more difficult to justify the assumption (\ref{eq:ass-main}), especially if the external potential is singular, as it is in (\ref{eq:Hatom-N}). Similarly, let us stress the fact that Theorem \ref{thm:main} remains true if in (\ref{eq:ham-mf}) and in (\ref{eq:HF}) we replace the repulsive Coulomb potential $V(x) = 1/|x|$ with the attractive interaction $V(x) = - 1/|x|$. Also here, however, it may be more difficult to justify (\ref{eq:ass-main}) in the attractive case.   

Finally, let us add a remark concerning the convergence of the higher order reduced densities. Theorem \ref{thm:main} only establishes the convergence of the one-particle reduced density. It turns out that our method can be extended to show the convergence of the $k$-particle reduced density, for any fixed $k \in \bN$, but only when tested against observables that are diagonal in a basis of $L^2 (\bR^{3k})$ consisting of factorized functions.

\section{Fock Space Representation}

To prove Theorem \ref{thm:main} we switch to a Fock space representation of the fermionic system. The fermionic Fock space over $L^2\left( \mathbb{R}^3\right)$ is defined as the direct sum
\begin{align}
\notag
\mathcal{F}= \bigoplus_{n\geq0} L_a^2\left( \mathbb{R}^{3n}\right),
\end{align}
where $L_a^2\left( \mathbb{R}^{3n}\right)$ is the antisymmetric subspace of $L^2\left( \mathbb{R}^{3n}\right)$. 

The number of particle operator on $\cF$ is the closure of the symmetric operator defined by $(\cN \Psi)^{(n)} = n \psi^{(n)}$ for all $\Psi = \{ \psi^{(n)} \}_{n \geq 0} \in \cF$ with $\psi^{(n)} = 0$ for all $n$ large enough. 

On $\cF$, it is useful to introduce creation and annihilation operators. For $f \in L^2\left( \mathbb{R}^{3}\right)$, we define the creation operator $a^* (f)$ and the annihilation operator $a(f)$ through 
\[ \begin{split} 
\left(a^*(f)\Psi\right)^{(n)}(x_1,\ldots, x_n) & := \frac{1}{\sqrt{n}} \sum_{j=1}^n (-1)^j f(x_j) \\ & \hspace{2cm} \times  \psi^{(n-1)}(x_1, \ldots,x_{j-1}, x_{j+1}, \ldots, x_n)\,, \\
\left(a(f)\Psi\right)^{(n)}(x_1,\ldots, x_n)&:= \sqrt{n+1} \int \rd x \; \overline{f}(x) \, \psi^{(n+1)} (x, x_1,\ldots, x_n)\,,
\end{split}\] 
for all $\Psi = \{ \psi^{(n)} \}_{n \geq 0} \in \cF$ with $\psi^{(n)} = 0$ for all $n$ large enough. Creation and annihilation operators satisfy canonical anticommutation relations 
\begin{align}
\left\lbrace a(f), a^*(g) \right\rbrace = \langle f, g \rangle, \quad \left\lbrace a(f), a(g) \right\rbrace = \left\lbrace a^*(f), a^*(g) \right\rbrace =0
\label{eq:CAR} 
\end{align}
for all $f,g \in L^2 \left( \mathbb{R}^3\right)$. Using the anticommutation relations it is easy to see that $a(f)$ and $a^* (f)$ extend to bounded operators on $\cF$, with $\| a (f) \| = \| a^* (f) \| = \| f \|_2$, and that $a^* (f)$ is the adjoint of $a(f)$. 

It is also convenient to define operator valued distributions $a_x^* , a_x$ for $x\in \bR^3$, such that 
\[ a^* (f) = \int dx \, f(x) a_x^* , \quad a(f) = \int dx \, \overline{f} (x) a_x\,. \]
In terms of the distributions $a_x^*, a_x$, we find 
\[ \cN = \int dx \, a_x^* a_x\,. \]

More generally, for an operator $J$ on $L^2 (\bR^3)$ we define its second quantization $d\Gamma (J)$ so that its restriction to the $n$-particle sector has the form
\[ d\Gamma (J) |_{\cF_n} = \sum_{j=1}^n J^{(j)} \]
where $J^{(j)} = 1^{\otimes (n-j)} \otimes J \otimes 1^{\otimes (j-1)}$ acts non-trivially only on the $j$-th particle. If $J$ has the integral kernel $J (x;y)$, we can write $d\Gamma (J)$ in terms of the distributions $a_x^*, a_x$ as  
\[ d\Gamma (J) = \int dx dy \, J (x;y) a_x^* a_y\,. \]
For example, $d\Gamma (1) = \cN$. In the next lemma we collect some bounds for the second quantization of one-particle operators. Its proof can be found in \cite[Lemma 3.1]{BPS}
\begin{lemma}
\label{lemma:bounds}
For every bounded operator $J$ on $L^2 \left( \mathbb{R}^3 \right)$, we have 
\[ \begin{split}
\langle \Psi , d\Gamma (J) \Psi \rangle &\leq \| J \| \langle \Psi, \cN \Psi \rangle\,, \\ 
\left\| \rd \Gamma (J) \Psi \right\| &\leq \left\| J \right\| \; \left\| \mathcal{N}\Psi\right\|,
\end{split}\]
for every $\Psi \in \mathcal{F}$. If $J$ is a Hilbert-Schmidt operator, we also have the bounds 
\[ \begin{split} 
\left\| \rd \Gamma (J) \psi \right\| &\leq \left\| J \right\|_{\mathrm{HS}} \left\| \mathcal{N}^{1/2} \psi \right\| \\
\left\| \int \rd x \; \rd y \; J (x; x') a_x a_{x'} \psi \right\| & \leq \| J \|_{\mathrm{HS}} \left\| \mathcal{N}^{1/2} \Psi \right\| \\
\left\| \int \rd x \; \rd y \; J (x; x') a^*_x a^*_{x'} \psi \right\| & \leq \| J \|_{\mathrm{HS}} \left\| \mathcal{N}^{1/2} \Psi \right\| 
\end{split}\]
for every $\Psi \in \mathcal{F}$. Finally, if $J$ is a trace class operator, we obtain 
\begin{align*}
\left\| \rd \Gamma (J) \right\| &\leq 2  \left\| J \right\|_{\mathrm{tr}} \\
\left\| \int \rd x \; \rd y \; J (x; x') a_x a_{x'} \right\| & \leq 2  \left\| J \right\|_{\mathrm{tr}} \\
\left\| \int \rd x \; \rd y \;  J (x; x') a^*_x a^*_{x'}  \right\| & \leq 2  \left\| J \right\|_{\mathrm{tr}} 
\end{align*}
where $\| J \|_{\mathrm{tr}} = \tr |J| = \tr \sqrt{J^*J}$ indicates the trace norm of $J$.
\end{lemma}

For a Fock space vector $\Psi \in \cF$, we can define the one-particle reduced density as the non-negative trace class operator on $L^2 (\bR^3)$ with integral kernel 
\begin{equation}\label{eq:gamma-FS} \gamma_\Psi^{(1)} (x;y) = \langle \Psi, a_y^* a_x \Psi \rangle\,. \end{equation}
For a $N$ particle state $\Psi = \{ 0, 0 , \dots , \psi_N, \dots \} \in \cF$, it is easy to check that this definition coincides with (\ref{eq:gamma1}). In fact
\[ \begin{split}  
\langle \Psi, a_y^* a_x \Psi \rangle &= \langle a_y \Psi, a_x \Psi \rangle \\  &= \langle (a_y \Psi)^{(N-1)}, (a_x \Psi)^{(N-1)} \rangle \\ &= N \int dx_1 \dots dx_{N-1} \psi_N (x,x_1, \dots , x_{N-1}) \overline{\psi}_N (y, x_1, \dots , x_{N-1})\,. \end{split} \]
Furthermore, for a one-particle observable $J$ on $L^2 (\bR^3)$, we find that the expectation of the second quantization of $J$ in the Fock state $\Psi$ is given by 
\[ \langle \Psi, d\Gamma (J) \Psi \rangle = \int dx dy \, J (x;y) \, \langle \Psi, a_x^* a_y \Psi \rangle = \tr \, J \gamma^{(1)}_\Psi\,. \]
This motivates the definition (\ref{eq:gamma-FS}). Notice in particular, that with this definition
\[ \tr \, \gamma^{(1)}_\Psi = \langle \Psi , \cN \Psi \rangle \]
is the expected number of particles in $\Psi$. 

Next, we introduce the Hamilton operator $\cH_N$ on the fermionic Fock space $\cF$. Formally, we define $\cH_N$ in terms of the distributions $a_x^*, a_x$ as
\begin{equation}\label{eq:FocKHN} \cH_N = \eps^2 \int dx \,  \nabla_x  a_x^* \nabla_x a_x + \frac{1}{2N} \int dx dy \, V(x-y) a_x^* a_y^* a_y a_x\,.
\end{equation}
More precisely, $\cH_N$ is the self-adjoint operator whose restriction on the $n$-particle sector of $\cF$ is given by 
\[ \cH_N |_{\cF_n} = \sum_{j=1}^n -\eps^2 \Delta_{x_j} + \frac{1}{N} \sum_{i<j}^n V (x_i - x_j)\,. \]
In particular, when restricted on $\cF_N$, the Hamilton operator $\cH_N$ coincide with the mean field Hamilton operator (\ref{eq:ham-mf}) defined in the previous section (and thus, for initial data in $\cF$ with exactly $N$ particles, the dynamics generated by $\cH_N$ coincides exactly with the evolution introduced in the previous section).

Let $\{ f_j \}_{j=1}^N$ be an orthonormal system in $L^2 (\bR^3)$. On $\cF$, we consider the Slater determinant
\[ a^* (f_1) \dots a^* (f_N) \Omega = \, \left\{ 0, \dots, 0, \frac{1}{\sqrt{N!}}  \det (f_i (x_j))_{i,j \leq N}, 0, \dots \right\}\,, \]
where the only non-trivial entry is in the $N$-particle sector. As stated in (\ref{eq:omega0}), the one-particle reduced density associated to this Slater determinant is given by the orthogonal projection 
\[ \omega_N = \sum_{j=1}^N |f_j \rangle \langle f_j|\,. \]
An important observation is the fact that there exists a unitary operator, that we will denote by $R_{\omega_N} : \cF \to \cF$ with the following two properties: 
\[ R_{\omega_N} \Omega = a^* (f_1) \dots a^* (f_N) \Omega \]
and 
\begin{equation}\label{eq:bog} R^*_{\omega_N} a (g) R_{\omega_N} = a (u_N g) + a^* (\overline{v}_N \overline{g})
\end{equation}
where $u_N = 1- \omega_N$ and $v_N = \sum_{j=1}^N | \overline{f}_j \rangle \langle f_j|$. In other words, if we complete the orthonormal system $\{ f_j \}_{j=1}^N$ to an orthonormal basis $\{ f_j \}_{j \geq 1}$ of $L^2 (\bR^3)$, we find 
\[ R^*_{\omega_N} a(f_j) R_{\omega_N} = a (f_j), \quad \text{and }  \quad R^*_{\omega_N} a^* (f_j) R_{\omega_N} = a^* (f_j) \]
if $j > N$ while 
\[ R^*_{\omega_N} a(f_j) R_{\omega_N} = a^* (\overline{f}_j), \quad \text{and }  \quad  R^*_{\omega_N} a^* (f_j) R_{\omega_N} = a (\overline{f}_j) \]
if $j \leq N$. The map $R_{\omega_N}$ is known as a particle-hole transformation. It let us switch to a new representation of the system; the new vacuum describes the Slater determinant with reduced density $\omega_N$. The new creation operators create excitations of the Slater determinant, i.e. either particles outside the determinant or holes in it. The proof of the existence of the unitary operator $R_{\omega_N}$ with the properties listed above can be found, for example, in \cite{Solovej}. 

Theorem \ref{thm:main} is a consequence of the following theorem for the evolution of approximate Slater determinants in the Fock space $\cF$.
\begin{theorem}\label{thm:FS}
Let $\omega_N$ be a sequence of orthogonal projections on $L^2 (\bR^3)$, with $\tr \, \omega_N = N$ and $\tr\, (-\eps^2 \Delta) \, \omega_N \leq C N$. Let $\omega_{N,t}$ denote the solution of the Hartree-Fock equation (\ref{eq:HF}) with initial data $\omega_{N,0} = \omega_N$.  We assume that there exists $T > 0$, $p > 5$ and $C > 0$ such that 
\begin{equation}\label{eq:assFS} \sup_{t \in [0;T]} \sum_{i=1}^3 \left[ \| \rho_{|[x_i, \omega_{N,t}]|} \|_1 +  \| \rho_{|[x_i,\omega_{N,t}]|} \|_p \right]  \leq C N\eps .
\end{equation}

Let $\xi_N \in \cF$ be a sequence with 
\[ \langle \xi_N, \cN \xi_N \rangle \leq C N^\alpha \]  
for an exponent $\alpha$, with $0 \leq \alpha < 1$.  
We consider the evolution 
\[ \Psi_{N,t} = e^{-i\cH_N t/\eps} R_{\omega_N} \xi_N \]
and denote by $\gamma^{(1)}_{N,t}$ the one-particle reduced density of $\Psi_{N,t}$, as defined in (\ref{eq:gamma-FS}). Then for all $\delta > 0$ there is a constant $C > 0$ such that 
\[ \sup_{t \in [0;T]} \left\| \gamma_{N,t}^{(1)} - \omega_{N,t} \right\|_\text{HS} \leq C \left[ N^{\alpha/2} + N^{5/12 + \delta} \right] \]
and 
\[ \sup_{t \in [0;T]} \tr \, \left| \gamma_{N,t}^{(1)} - \omega_{N,t} \right| \leq C \left[ N^{\alpha} + N^{11/12 + \delta}\right] \, .\]
\end{theorem}

Let us show how Theorem \ref{thm:FS} implies the statement of Theorem \ref{thm:main}, where we consider the evolution of $N$-particle states. 

\begin{proof}[Proof of Theorem \ref{thm:main}]
Set $\Psi_N = \{ 0, \dots , 0, \psi_N, 0, \dots \}$ and $\xi_N = R_{\omega_N}^* \Psi_N \in \cF$. Then we have $\Psi_N = R_{\omega_N} \xi_N$, and 
\[ \begin{split} \langle \xi_N, \cN \xi_N \rangle &= \langle R_{\omega_N}^* \Psi_N, \cN R_{\omega_N}^* \Psi_N \rangle \\ &= \int dx \, \langle \Psi_N, \left( a^* (u_x) + a (\overline{v}_x) \right) \left( a (u_x) + a^* (\overline{v}_x) \right) \Psi_N \rangle\,.\end{split} \]
Using the anticommutation relations we find $a(\overline{v}_x) a^* (\overline{v}_x) = - a^* (\overline{v}_x) a(\overline{v}_x) + \langle \overline{v}_x , \overline{v}_x \rangle$. Since $u_N = 1- \omega_N$ is orthogonal to $\omega_N$, we conclude that  
\[\begin{split} \langle \xi_N, \cN \xi_N \rangle &= \langle \Psi_N, \left[ d\Gamma (u_N) - d\Gamma (\omega_N) + N \right] \Psi_N \rangle \\ &= 2 \, \tr \, \gamma_N^{(1)} \left( 1- \omega_N \right)= 2 \, \tr (\gamma_N^{(1)} - \omega_N) (1-\omega_N)\,.\end{split} \]
This implies that 
\[ \langle \xi_N , \cN \xi_N \rangle \leq 2 \, \tr \, \left| \gamma_N^{(1)} - \omega_N \right| \leq C N^\alpha \]
for an exponent $0 \leq \alpha < 1$, from the assumption (\ref{eq:conden}). Hence, we can apply Theorem \ref{thm:FS} and we obtain that  
\[ \sup_{t \in [0;T]} \, \left\| \gamma^{(1)}_{N,t} - \omega_{N,t} \right\|_\text{HS} \leq C \left[ N^{\alpha/2} + N^{5/12+\delta} \right] \]
and that 
\[ \sup_{t \in [0;T]} \, \tr \left| \gamma^{(1)}_{N,t} - \omega_{N,t} \right|  \leq C \left[ N^\alpha + N^{11/12 + \delta} \right] \]
for any $\delta > 0$.
\end{proof}

In order to prove Theorem \ref{thm:FS}, we define the fluctuation dynamics
\begin{equation}\label{eq:fluc} \cU_N (t) = R_{\omega_{N,t}}^* e^{-i\cH_N t/\eps} R_{\omega_{N}} 
\end{equation}
and we observe that 
\[ \Psi_{N,t} = e^{-i\cH_N t/\eps} R_{\omega_N} \xi_N = R_{\omega_{N,t}} \cU_N (t) \xi_{N}\,. \]
The vector $\cU_N (t) \xi_N$ describes the excitations at time $t$. The key step in the proof of Theorem \ref{thm:FS} is the following bound on the expectation of the operator $\cN$ in the state $\cU_N (t) \xi_{N}$. This is a bound on the expected number of excitations of the Slater determinant in the state $\Psi_{N,t}$. 
\begin{proposition}\label{prop:Gronwall}
Let $\omega_N$ be a sequence of orthogonal projections on $L^2 (\bR^3)$, with $\tr \, \omega_N = N$ and $\tr\, (-\eps^2 \Delta) \, \omega_N \leq C N$. Suppose that there exists $T >0$, $p > 5$ and $C > 0$ such that
\begin{equation}\label{eq:assFS-prop}
 \sup_{t \in [0;T]} \sum_{i=1}^3 \left[ \| \rho_{|[x_i, \omega_{N,t}]|} \|_1 +  \| \rho_{|[x_i,\omega_{N,t}]|} \|_p \right]  \leq C N\eps .
\end{equation}
Let $\cU_N (t)$ be the fluctuation dynamics defined in (\ref{eq:fluc}) and $\xi_N \in \cF$. Then, for every $\delta > 0$ small enough, there exists a constant $C > 0$ such that 
\[ \sup_{t \in [0;T]} \, \langle \cU_N (t) \xi_N, \cN \cU_N (t) \xi_N \rangle \leq C \left[ \langle \xi_N , \cN \xi_N \rangle + N^{5/6+\delta} \right] \, , \]
for all $\xi_N \in \cF$ with $\| \xi_N \| = 1$. 
\end{proposition}

Remark that the proof of this proposition, which will be given in the next section, can be extended to show a similar bound for higher moments of the number of particles operator (these estimates are needed to establish the convergence of higher order reduced densities, as stated after Theorem \ref{thm:main}). Let us now show how Proposition \ref{prop:Gronwall} can be used to establish Theorem \ref{thm:FS}.

\begin{proof}[Proof of Theorem \ref{thm:FS}] We follow here the same argument used in \cite{BPS}. {F}rom (\ref{eq:gamma-FS}), we obtain 
\[ \begin{split} 
\gamma_{N,t}^{(1)}(x;y) &= \left\langle \Psi_{N,t}, a_y^*a_x \Psi_{N,t}\right\rangle \\
&= \left\langle  \cU_N (t) \xi_N, R_{\omega_{N,t}}^* a_y^* a_x R_{\omega_{N,t}} \cU_N (t) \xi_N\right\rangle.
\end{split} \]
Eq. (\ref{eq:bog}) implies that 
\begin{equation}\label{eq:22pr}
\begin{split}
\gamma_{N,t}^{(1)}(x;y) &= \left\langle  \cU_N (t;0) \xi_N, \left( a^*\left( u_{t,y}\right) + a\left(\bar{v}_{t,y}\right)\right) \left( a\left( u_{t,x}\right) + a^*\left(\bar{v}_{t,x}\right)\right)  \cU_N (t;0) \xi_N \right\rangle \\ 
&= \left\langle \UN  \xi_N, \left[ a^*\left( u_{t,y}\right) a\left( u_{t,x}\right) - a^*\left(\bar{v}_{t,x}\right) a\left(\bar{v}_{t,y}\right)+ \left\langle \bar{v}_{t,y}, \bar{v}_{t,x}\right\rangle \right. \right. \\
&\hspace{2cm} \left. \left.  +  a^*\left( u_{t,y}\right) a^*\left(\bar{v}_{t,x}\right)+  a\left(\bar{v}_{t,y}\right) a\left( u_{t,x}\right)  \right] \UN \xi_N\right\rangle ,
\end{split}
\end{equation}
where we introduced the short-hand notation $u_{t,x} (z) = u_{N,t} (x;z)$ and $\overline{v}_{t,y} (z) = \overline{v}_{N,t} (y;z)$  for the kernels of $u_{N,t} = 1-\omega_{N,t}$ and $v_{N,t} = \sum_{j=1}^N |\overline{f}_j \rangle \langle f_j|$ if $\omega_{N,t} = \sum_{j=1}^N |f_j \rangle \langle f_j|$. Notice that then 
\[ \left\langle \bar{v}_{t,y}, \bar{v}_{t,x}\right\rangle = \int \rd z \; v_{N,t} \left(z;y\right) \bar{v}_{N,t} (z;x)= \left( v_{N,t} \bar{v}_{N,t} \right)\left(y;x\right) = \omega_{N,t}\left(x;y\right) \, .\]
This leads to 
\begin{align}
\gamma_{N,t}^{(1)}\left(x;y\right) &-\omega_{N,t}\left(x;y\right) \\ &= \left\langle \UN  \xi_N, \left[ a^*\left( u_{t,y}\right) a\left( u_{t,x}\right) - a^*\left(\bar{v}_{t,x}\right) a\left(\bar{v}_{t,y}\right)\right. \right. \notag\\
&\hspace{1cm} \left. \left.  +  a^*\left( u_{t,y}\right) a^*\left(\bar{v}_{t,x}\right)+  a\left(\bar{v}_{t,y}\right) a\left( u_{t,x}\right)  \right] \UN \xi_N\right\rangle.\label{eq:gamma-w}
\end{align}

Let $J$ be a Hilbert-Schmidt operator on $L^2 \left( \mathbb{R}^3 \right) $. Integrating its kernel against the difference \eqref{eq:gamma-w}, we find
\begin{equation}\label{eq:trOdiff} 
\begin{split} \tr \; J &\left( \gamma_{N,t}^{(1)}-\omega_{N,t}\right) \\ =& \left\langle \xi_N, \UN^* \left( \rd \Gamma \left( u_{N,t} J u_{N,t} \right) - \rd \Gamma \left( \bar{v}_{N,t}\overline{J^*}v_{N,t}\right) \right) \UN \xi_N \right\rangle \\
&+ 2\Re \big\langle\xi_N, \UN^* \big( \int dr_1 dr_2 \,\left( v_{N,t} J u_{N,t} \right) \left(r_1,r_2\right) a_{r_1}a_{r_2} \big) \UN \xi_N \big\rangle.
\end{split} \end{equation}
Using Lemma \ref{lemma:bounds} and $\|u_{N,t}\|=\|v_{N,t}\|=1$, we find 
\begin{equation*}
\begin{split}
\big\vert \tr \; J \, \big( \gamma_{N,t}^{(1)} &-\omega_{N,t}\big) \big\vert \\ \leq& \left( \|u_{N,t} J u_{N,t}\|_{\mathrm{HS}}+\|\bar{v}_{N,t}\overline{J^*}v_{N,t}\|_{\mathrm{HS}}\right)\left\| \left(\mathcal{N}+1\right)^{1/2} \UN \xi_N \right\|   \\
&\quad + 2 \; \| v_{N,t} J u_{N,t}  \|_{\mathrm{HS}} \left\| \left(\mathcal{N}+1\right)^{1/2} \UN \xi_N \right\| \\ 
\leq& C \| J \|_{\mathrm{HS}} \left\| \left(\mathcal{N}+1\right)^{1/2} \UN \xi_N \right\|. 
\end{split} 
\end{equation*}
By duality, this implies that 
\[ \| \gamma_{N,t}^{(1)} - \omega_{N,t} \|_\text{HS} \leq C \left\| \left(\mathcal{N}+1\right)^{1/2} \UN \xi_N \right\|\,. \]
With Proposition \ref{prop:Gronwall} we conclude that 
\[ \sup_{t \in [0;T]} \| \gamma_{N,t}^{(1)} - \omega_{N,t} \|_\text{HS} \leq C \left[ N^{\alpha/2} + N^{5/12 + \delta} \right] \]
for any $\delta > 0$. 

Finally, we prove the trace class bound (\ref{eq:tr-bd}). Starting from (\ref{eq:trOdiff}) we find, for any compact operator $J$ on $L^2 (\bR^3)$, 
\begin{align*}
\big\vert \tr \; J \, \big( \gamma_{N,t}^{(1)}- &\omega_{N,t}\big) \big\vert \\ \leq& \left( \|u_{N,t} J u_{N,t}\|+\|\bar{v}_{N,t}\overline{J^*}v_{N,t}\| \right) \left\langle \xi_N, \UN^* \mathcal{N}\UN \xi_N\right\rangle \notag\\
&\quad + 2 \; \| v_{N,t} J u_{N,t}  \|_{\mathrm{HS}} \left\| \left(\mathcal{N}+1\right)^{1/2} \UN \xi_N \right\|  \; \| \xi_N \|\notag\\
\leq& 2 \|J\|\left\| \left(\mathcal{N}+1\right)^{1/2}\UN \xi_N\right\|^2 \notag\\
&\quad + 2 \; \| v_{N,t} J u_{N,t}  \|_{\mathrm{HS}} \left\| \left(\mathcal{N}+1\right)^{1/2} \UN \xi_N \right\|  \; \| \xi_N \|\notag\\
\leq& 2 \|J\|\left\| \left(\mathcal{N}+1\right)^{1/2}\UN \xi_N\right\| \notag\\
&\quad + 2 \; \| J \| \; \| v_{N,t} \|_{\mathrm{HS}} \left\| \left(\mathcal{N}+1\right)^{1/2} \UN \xi_N \right\|  \; \| \xi_N \|.\notag
\end{align*}
From Proposition \ref{prop:Gronwall} and $\| v_{N,t} \|_{\mathrm{HS}} \leq N^{\frac12}$, we obtain that, for every $\delta > 0$ there exists $C > 0$ such that 
\begin{align}
\tr \; \left| \gamma_{N,t}^{(1)}-\omega_{N,t}\right| \leq C \left[ N^{\alpha} + N^{11/12+\delta} \right] \, . 
\end{align}
This completes the proof of Theorem \ref{thm:FS}.
\end{proof}

\section{Control of the Fluctuations}

The goal of this section is to show Proposition \ref{prop:Gronwall}. 
To reach this goal, we derive a differential inequality for the expectation $\langle \cU_N (t) \xi_N , \cN \cU_N (t) \xi_N \rangle$ and we apply Gronwall's lemma. We have 
\begin{equation}\label{eq:ddt}
\begin{split}
i \varepsilon \frac{\rd}{\rd t}  &\left\langle \mathcal{U}_N\left(t;0\right) \xi_N , \mathcal{N} \mathcal{U}_N\left(t;0\right)\xi_N \right\rangle \\
= \; &\frac{4 \ri}{N} \mathrm{Im} \int \rd x \, \rd y \; \frac{1}{|x-y|} \\
&\times \bigg\lbrace \left\langle \UN \xi_N , a^*\left(u_{t,x}\right)a\left(\bar{v}_{t,y}\right) a\left(u_{t,y}\right)a\left(u_{t,x}\right) \UN \xi_N \right\rangle \\
& \hspace{.5cm} +  \left\langle \UN \xi_N , a^*\left(u_{t,y}\right)a^*\left(\bar{v}_{t,y}\right) a^*\left(\bar{v}_{t,x}\right)a\left(\bar{v}_{t,x}\right) \UN \xi_N \right\rangle \\
& \hspace{.5cm} + \left\langle \UN \xi_N, a\left(\bar{v}_{t,x}\right)a\left(\bar{v}_{t,y}\right) a\left(u_{t,y}\right)a\left(u_{t,x}\right) \UN \xi_N \right\rangle \bigg\rbrace, 
\end{split}
\end{equation}
where, as in the last section, we use the short-hand notation $u_{t,x} (z) = u_{N,t} (x;z)$, $v_{t,x} (z) = v_{N,t} (x;z)$, with the operators $u_{N,t} = 1-\omega_{N,t}$ and $v_{N,t}$ as defined after (\ref{eq:22pr}). The proof of (\ref{eq:ddt}) is a lengthy but straightforward  computation that can be found in \cite[Proof of Proposition 3.3]{BPS}.

Next, we estimate the three contribution on the r.h.s. of (\ref{eq:ddt}) separately. We start with the term 
\begin{equation}\label{eq:I} \text{I} = \frac{1}{N} \int dx dy \frac{1}{|x-y|} \langle \cU_N (t;0) \xi_N, a^* (u_{t,x}) a (\overline{v}_{t,y}) a (u_{t,y}) a(u_{t,x}) \, \cU_N (t;0) \xi_N \rangle\,. 
\end{equation}

To bound this contribution (and later also to control the other two terms on the r.h.s of (\ref{eq:ddt})), we use a smooth version of the Fefferman-de la Llave representation of the Coulomb potential  \cite{FDL}, given by 
\begin{equation}\label{eq:ff} \frac{1}{|x-y|} = \frac{4}{\pi^2}  \int_0^\infty \frac{dr}{r^5} \int dz \, \chi_{(r,z)} (x) \chi_{(r,z)} (y) 
\end{equation}
where we introduced the notation $\chi_{(r,z)} (x) = e^{-(x-z)^2/r^2}$. The proof of (\ref{eq:ff}) is a simple computation with Gaussian integrals which we leave to the reader (the fact that the result of the integral is proportional to $|x-y|^{-1}$, which is the only property we are going to use, follows by simple scaling). Inserting (\ref{eq:ff}) into (\ref{eq:I}) we find 
\begin{equation}\label{eq:I1} \begin{split}
\text{I} = \; &\frac{C}{N} \int dx dy \; \int_0^\infty \frac{\rd r}{r^5} \int \rd z \; \chi_{\left(r,z\right)}(x) \chi_{\left(r,z\right)}(y)  \\ &\hspace{1cm} \times \langle \UN \xi_N,  a^*\left(u_{t,x}\right) a\left( \bar{v}_{t,y}\right) a\left(u_{t,y}\right) a\left( u_{t,x}\right)\UN \xi_N \rangle \\ 
=\; & \frac{C}{N} \int_0^\infty \frac{\rd r}{r^5} \int \rd z \; \rd x\; \chi_{\left(r,z\right)}(x) \\ &\hspace{1cm} \times \langle \UN \xi_N ,  a^*\left(u_{t,x}\right) B_{r,z} a\left( u_{t,x}\right)\UN \xi_N\rangle , 
\end{split}
\end{equation}
where we defined the operator \begin{equation}\label{eq:Brz} B_{r,z} = \int \rd y \; a\left( \bar{v}_{t,y}\right) \chi_{\left(r,z\right)}(y)  a\left(u_{t,y}\right) = \int ds_1 ds_2 (\overline{v}_{N,t} \chi_{(r,z)} u_{N,t}) (s_1; s_2) a_{s_1} a_{s_2}\,.
\end{equation}
Lemma \ref{lemma:bounds} implies that 
\begin{equation}\label{eq:normB} \| B_{r,z} \| \leq 2\| \overline{v}_{N,t} \chi_{(r,z)} u_{N,t} \|_\text{tr} \leq 2\left\| \left[ \chi_{(r,z)} , \omega_{N,t} \right] \right\|_\text{tr}\,. \end{equation}
To bound the r.h.s., we use the next lemma, whose proof is deferred to the end of the section.
\begin{lemma}
\label{lm:tr1}
Let $\chi_{r,z} (x) = \exp (-x^2/r^2)$. Then, for all $0 < \delta <1/2$ there exists $C > 0$ such that the pointwise bound
\begin{equation}\label{eq:tr1} \left\| [\chi_{(r,z)} , \omega_{N,t} ] \right\|_\text{tr} 
\leq C \, r^{\frac{3}{2} - 3\delta} \sum_{i=1}^3 \| \rho_{|[x_i, \omega_{N,t}]|} \|_1^{\frac{1}{6}+\delta} \left( \rho^*_{|[x_i, \omega_{N,t}]|} (z) \right)^{\frac{5}{6} - \delta}
\end{equation}
holds true. Here $\varrho^*_{|[x_i,\omega_{N,t}]|}$ denotes the Hardy-Littlewood maximal function defined by 
\begin{equation}\label{eq:max-def} \rho^*_{|[x_i , \omega_{N,t}]|} (z) = \sup_{B : z \in B} \frac{1}{|B|} \int_B dx\,\rho_{|[x_i , \omega_{N,t}]|} (x)  \end{equation}
with the supremum taken over all balls $B \in \bR^3$ such that $z \in B$. 
\end{lemma}

Applying (\ref{eq:tr1}) to the r.h.s. of (\ref{eq:normB}) and using the assumption (\ref{eq:assFS-prop}), we conclude from (\ref{eq:I1}) that, for all $\delta > 0$ there exists $C > 0$ such that 
\[ \begin{split} 
|\text{I}| &\leq C \frac{(N\eps)^{1/6+\delta}}{N} \sum_{i=1}^3  \int_0^\infty \frac{\rd r}{r^{7/2+3\delta}} \; \int\rd x  \rd z \;  \chi_{\left(r,z\right)} (x)  \; \left(\varrho^*_{|[x_i,\omega_{N,t}]|}(z)\right)^{\frac{5}{6}-\delta} \\ &\hspace{6cm} \times  \| a\left( u_{t,x}\right)\UN \xi_N \|^2  \\
&\leq C \frac{(N\eps)^{1/6+\delta}}{N} \sum_{i=1}^3 \int_0^\infty \frac{\rd r}{r^{7/2+3\delta}} \; \int\rd x  \, g_{i,r} (x) \, \| a\left( u_{t,x}\right)\UN \xi_N \|^2 \end{split} \]
where we defined \begin{equation}\label{eq:grdef} g_{i,r} (x) = \int dz \, \chi_{(r,z)} (x) \left(\varrho^*_{|[x_i,\omega_{N,t}]|}(z)\right)^{5/6-\delta}\,. 
\end{equation}
We find
\[ \begin{split} 
|\text{I}| \leq C \frac{(N\eps)^{1/6+\delta}}{N}  \sum_{i=1}^3 \int_0^\infty &\frac{\rd r}{r^{7/2+3\delta}}  \\ &\times \langle \UN \xi_N ,  d\Gamma (u_{N,t} g_{i,r} (x) u_{N,t}) \, \UN \xi_N \rangle 
\end{split} \]
where $u_{N,t} g_{i,r} (x) u_{N,t}$ is the operator with the integral kernel 
\[ (u_{N,t} g_{i,r} (x)u_{N,t}) (s_1; s_2) = \int dx\,u_{N,t} (s_1 ;x) g_{i,r} (x) u_{N,t} (x;s_2)\,. \]
Applying again Lemma \ref{lemma:bounds} and using the fact that $\| u_{N,t} \| \leq 1$, we obtain
\begin{equation}\label{eq:Ibd} 
|\text{I}| \leq C  \frac{(N\eps)^{1/6+\delta}}{N}  \sum_{i=1}^3 \int_0^\infty \frac{dr}{r^{7/2+3\delta}} \, \| g_{i,r} \|_\infty \| \cN^{1/2} \UN \xi_N \|^2\,. \end{equation}
We have, using the Hardy-Littlewood maximal inequality, 
\begin{equation}\label{eq:grin} \| g_{i,r} \|_\infty \leq r^{3/p} \, \| \varrho^*_{|[x_i,\omega]|} \|_{(5/6-\delta)q}^{5/6-\delta} \leq C r^{3/p} \| \varrho_{|[x_i,\omega]|} \|_{(5/6-\delta)q}^{5/6-\delta} \end{equation}
for any $(5/6-\delta)^{-1} <  q \leq \infty$ and $p$ such that $p^{-1} + q^{-1} = 1$. To bound the r.h.s. of (\ref{eq:Ibd}), we divide the $r$-integral into two parts and then we apply (\ref{eq:grin}) with two different choices of $p,q$. {F}rom the assumption (\ref{eq:assFS-prop}) we can find $q_1 > 6$ and $q_2 < 6$ and $\delta > 0$ sufficiently small such that  
\[ \sup_{t \in [0;T]} \| \rho_{|[x_i,\omega_{N,t}]|} \|_{q_1(5/6-\delta)} + \| \rho_{|[x_i,\omega_{N,t}]|} \|_{q_2(5/6-\delta)} \leq C N\eps\,.  \]
With this choice of $q_1$, $q_2$, we have $p_1 < 6/5$ and $p_2 > 6/5$ which implies (possibly after reducing again the value of $\delta > 0$) that $r^{-7/2 -3\delta + 3/p_1}$ is integrable close to zero and that $r^{-7/2-3\delta + 3/p_2}$ is integrable at infinity. We conclude that
\begin{equation}\label{eq:If} | \text{I} | \leq C \eps \, \| \cN^{1/2} \UN \xi_N \|^2 = C \eps \, \langle \UN \xi_N, \cN \UN \xi_N \rangle 
\end{equation}
for all $t \in [0;T]$. 

The second term on the l.h.s. of \eqref{eq:ddt} can be estimated similarly. Recalling the definition (\ref{eq:Brz}) of the operator $B_{r,z}$, we can write
\[ \begin{split} \text{II} &= \frac{C}{N} \int dx dy \frac{1}{|x-y|} \, \langle \UN \xi_N , a^* (\overline{v}_{t,x}) a^* (u_{t,y}) a^* (\overline{v}_{t,y}) a(\overline{v}_{t,x}) \UN \xi_N \rangle 
\\ &= \frac{C}{N} \int dx dy \int_0^\infty \frac{dr}{r^5} \int dz \chi_{r,z} (x) \chi_{r,z} (y) \\ &\hspace{2cm} \times \langle \UN \xi_N , a^* (\overline{v}_{t,x}) a^* (u_{t,y}) a^* (\overline{v}_{t,y}) a(\overline{v}_{t,x}) \UN \xi_N \rangle  \\
&= \frac{C}{N} \int dx \int_0^\infty \frac{dr}{r^5} \int dz \chi_{r,z} (x) \\ &\hspace{2cm} \times \langle \UN \xi_N , a^* (\overline{v}_{t,x})  B^*_{r,z} a (\overline{v}_{t,x}) \UN \xi_N \rangle \end{split} \]
which implies, with (\ref{eq:normB}), (\ref{eq:tr1}), the assumptions (\ref{eq:assFS-prop}) and (\ref{eq:grdef}) that, for $\delta >0$ small enough, 
\[ \begin{split} |\text{II}| &\leq \frac{C (N\eps)^{1/6+\delta}}{N} \sum_{i=1}^3  \int_0^\infty \frac{dr}{r^{7/2+3\delta}} \int dx dz \, \chi_{r,z} (x) \left( \rho^*_{|[x_i,\omega]|} (z) \right)^{5/6-\delta} \\ &\hspace{5cm} \times  \left\| a (\overline{v}_{t,x}) \UN \xi_N \right\|^2 \\
&\leq \frac{C (N\eps)^{1/6+\delta}}{N} \sum_{i=1}^3 \int_0^\infty \frac{dr}{r^{7/2+3\delta}}   \\ &\hspace{3cm} \times \langle \UN \xi_N, d\Gamma (\overline{v}_{N,t} g_{i,r} (x) \overline{v}_{N,t})  \UN \xi_N \rangle \\ 
&\leq \frac{C (N\eps)^{1/6+\delta}}{N} \sum_{i=1}^3  \int_0^\infty \frac{dr}{r^{7/2+3\delta}}   \| g_{i,r} \|_\infty \,  \left\| \cN^{1/2}  \UN \xi_N \right\|^2\,. 
 \end{split} \]
Then we conclude as we did for (\ref{eq:Ibd}) that 
\begin{equation}\label{eq:IIf} | \text{II} | \leq C \eps \, \langle \UN \xi_N , \cN \UN \xi_N \rangle 
\end{equation}
for all $t \in [0;T]$.

Finally, we consider the third term on the r.h.s. of  \eqref{eq:ddt}. Again we use the Fefferman-de la Llave formula \eqref{eq:ff} for the Coulomb potential. We obtain
\[\begin{split}
\mathrm{III} &= \frac{C}{N} \int \rd x\; \rd y \; \int_0^\infty \frac{\rd r}{r^5} \int \rd z \; \chi_{\left(r,z\right)}(x) \chi_{\left(r,z\right)}(y)  \\ &\hspace{1cm} \times  \langle \UN \xi_N,  a\left(\bar{v}_{t,x}\right) a\left( \bar{v}_{t,y}\right) a\left(u_{t,y}\right) a\left( u_{t,x}\right)\UN \xi_N \rangle\,. 
\end{split} \] 
We divide the $r$-intergral into two parts, setting $\text{III} = \text{III}_1 + \text{III}_2$, with 
\begin{equation}\label{eq:III12} \begin{split} \text{III}_1 =\; &\frac{C}{N} \int \rd x\; \rd y \; \int_0^\kappa \frac{\rd r}{r^5} \int \rd z \; \chi_{\left(r,z\right)}(x) \chi_{\left(r,z\right)}(y)  \\ &\hspace{.5cm} \times \langle \UN \xi_N,  a\left(\bar{v}_{t,x}\right) a\left( \bar{v}_{t,y}\right) a\left(u_{t,y}\right) a\left( u_{t,x}\right)\UN \xi_N \rangle\,, \\ \text{III}_2 =\; &\frac{C}{N} \int \rd x\; \rd y \; \int_\kappa^\infty \frac{\rd r}{r^5} \int \rd z \; \chi_{\left(r,z\right)}(x) \chi_{\left(r,z\right)}(y)  \\ &\hspace{.5cm} \times 
\langle \UN \xi_N,  a\left(\bar{v}_{t,x}\right) a\left( \bar{v}_{t,y}\right) a\left(u_{t,y}\right) a\left( u_{t,x}\right)\UN \xi_N \rangle\,.  \end{split} 
\end{equation}
We start estimating $\text{III}_1$. Here, we start by integrating over $z$. Since
\[ \int dz\, \chi_{(r,z)} (x) \chi_{(r,z)} (y)  = r^3 \chi_{(\sqrt{2} r, x)} (y) \]
we obtain, with (\ref{eq:Brz}), 
\[ \begin{split} \text{III}_1 &= \frac{C}{N} \int_0^\kappa \frac{dr}{r^2} \int dx dy \, \chi_{(\sqrt{2}r , x)} (y) \\ &\hspace{2cm} \times  \langle \UN \xi_N ,  a\left(\bar{v}_{t,x}\right) a\left( \bar{v}_{t,y}\right) a\left(u_{t,y}\right) a\left( u_{t,x}\right)\UN \xi_N \rangle \\ 
 &= \frac{C}{N} \int_0^\kappa \frac{dr}{r^2} \int dx  \, \langle \UN \xi_N, B_{\sqrt{2} r, x} a(\overline{v}_{t,x}) a (u_{t,x})
 \UN \xi_N \rangle\,.   \end{split} \]
Since $\| \overline{v}_{N,x} \|^2 = \omega_{N,t} (x;x) =: \rho_{N,t} (x)$, we find 
\begin{equation} 
\begin{split}
 | \text{III}_1 | &\leq \frac{C}{N} \int_0^\kappa \frac{dr}{r^2} \int dx  \, \| B_{\sqrt{2} r, x} \| \, \rho_{N,t}^{1/2} (x) \, \| a (u_{t,x}) \UN \xi_N \| \\
&\leq \frac{C}{N} \int_0^\kappa \frac{dr}{r^2} \int dx  \, \rho_{N,t}^{1/2} (x) \| [\chi_{(\sqrt{2}r,x)}, \omega_{N,t}] \|_\text{tr}    \| a (u_{t,x}) \UN \xi_N \|\,. 
\end{split} \end{equation}
Using the pointwise bound (\ref{eq:tr1}) and the assumption (\ref{eq:assFS-prop}), we obtain that, for all $\delta > 0$ sufficiently small, there exists a constant $C > 0$ such that 
\begin{equation}\label{eq:III1-a}
\begin{split} 
| \text{III}_1 | & \leq \frac{C (N\eps)^{\frac{1}{6} + \delta}}{N} \sum_{i=1}^3 \int_0^\kappa \frac{dr}{r^{1/2 +3\delta}} \int dx  \, \rho_{N,t}^{1/2} (x) \, \left[ \rho^*_{|[x_i, \omega_{N,t}]|} (x) \right]^{\frac{5}{6}-\delta}  \\ &\hspace{6cm} \times  \| a (u_{t,x}) \UN \xi_N \|\,.
\end{split} 
\end{equation}
Applying H\"older's inequality, we conclude that
\begin{equation}\label{eq:hold} \begin{split} 
| \text{III}_1 | & \leq \frac{C (N\eps)^{\frac{1}{6} + \delta} \kappa^{1/2 -3\delta}}{N} \sum_{i=1}^3 \| \rho_{N,t} \|_{5/3}^{1/2} \| \rho^*_{|[x_i, \omega_{N,t}]|} \|_{\frac{25}{6} - 5 \delta}^{\frac{5}{6} - \delta} \\ &\hspace{4cm} \times 
\left[ \int dx \, \| a(u_{t,x}) \UN \xi_N \|^2 \right]^{1/2}\,.
\end{split} \end{equation}
By the Hardy-Littlewood maximal inequality and the assumption (\ref{eq:assFS-prop}), we have
\begin{equation}\label{eq:HLmax} \| \rho^*_{|[x_i, \omega_{N,t}]|} \|_{\frac{25}{6} - 5 \delta} \leq  \| \rho_{|[x_i, \omega_{N,t}]|} \|_{\frac{25}{6} - 5 \delta} \leq C N\eps\,.
\end{equation}
Furthermore, we have 
\begin{equation}\label{eq:NUNxi} \begin{split} \int dx \, \| a(u_{t,x}) \UN \xi_N \|^2 &= \langle \UN \xi_N , d\Gamma (u_{N,t}) \UN \xi_N \rangle \\ &\leq \langle \UN \xi_N, \cN \UN \xi_N \rangle \\ &= \| \cN^{1/2} \UN \xi_N \|^2\,. \end{split} \end{equation}
On the other hand, to bound the norm $\| \rho_{N,t} \|_{5/3}$ we use the Lieb-Thirring inequality, which implies   
\[ \| \rho_{N,t} \|^{5/3}_{5/3} \leq \tr \, (-\Delta) \omega_{N,t} \leq \eps^{-2} \cE_\text{HF} (\omega_{N,t}) \]
with the Hartree-Fock energy
\[ \begin{split} \cE_\text{HF} (\omega_{N,t}) = \; & \tr \, (-\eps^2 \Delta) \omega_{N,t} \\ &+ \frac{1}{2N} \int \frac{1}{|x-y|} \left[ \omega_{N,t} (x;x) \omega_{N,t} (y;y) - |\omega_{N,t} (x;y)|^2 \right] dx dy\,. \end{split} \]
By energy conservation, we have 
\begin{equation}\label{eq:rhoNT} \| \rho_{N,t} \|^{5/3}_{5/3} \leq \eps^{-2} \cE_\text{HF} (\omega_{N})\,. \end{equation}
Next, we remark that the potential part of $\cE_\text{HF} (\omega_{N})$ can be bounded by its kinetic energy. In fact, applying the Hardy-Littlewood-Sobolev inequality and interpolation and using the normalization $\| \rho_N \|_1 = N$ for $\rho_N (x) = \omega_N (x;x)$, we find 
\[ \begin{split} \frac{1}{N} \int \frac{1}{|x-y|} \rho_{N} (x) \rho_{N} (y) dx dy &\leq \frac{C}{N} \| \rho_{N} \|_{6/5}^2 \\ &\leq \frac{C}{N} \| \rho_{N} \|_1^{7/5} \| \rho_{N} \|_{5/3}^{3/5}
\\ &= C N^{2/5} \| \rho_{N} \|_{5/3}^{3/5} \\ &\leq C N + C N^{-2/3} \| \rho_{N} \|_{5/3}^{5/3}\, , \end{split} \]
by Young's inequality. {F}rom the Lieb-Thirring, we find 
\[  \frac{1}{N} \int \frac{1}{|x-y|} \rho_{N} (x) \rho_{N} (y) dx dy \leq C N + C \tr \, (-\eps^2 \Delta) \omega_{N} \]
and hence
\[ \cE_\text{HF} (\omega_{N}) \leq C N +  C \tr \, (-\eps^2 \Delta) \omega_{N} \leq C N \]
from the assumption $\tr \, (-\eps^2 \Delta) \omega_N \leq CN$ on the initial sequence of orthogonal projection $\omega_N$. {F}rom (\ref{eq:rhoNT}), we conclude that $\| \rho_{N,t} \|_{5/3} \leq N$. Combining this estimate with (\ref{eq:HLmax}) and  (\ref{eq:NUNxi}), we obtain 
\[ | \text{III}_1 | \leq C \sqrt{N} \eps \kappa^{1/2 -3\delta} \| \cN^{1/2} \UN \xi_N \| \leq \eps \|  \cN^{1/2} \UN \xi_N \|^2 + C N \eps \kappa^{1-6\delta} \]
for all $t \in [0;T]$. 

Next, we estimate the second term in (\ref{eq:III12}). With the definition (\ref{eq:Brz}), we have 
\begin{equation*}  |\text{III}_2 | \leq \; \frac{C}{N} \int_\kappa^\infty \frac{\rd r}{r^5} \int \rd z \; \| B_{r,z} \|^2 \leq \frac{C}{N} \int_\kappa^{\infty} \frac{\rd r}{r^5} \int \rd z \;  \| [ \chi_{(r,z)} , \omega_{N,t}] \|_\text{tr}^2  \, . 
\end{equation*} 
With the bound (\ref{eq:tr1}) and the assumption (\ref{eq:assFS-prop}), we obtain
\[ |\text{III}_2 | \leq \frac{C (N\eps)^2}{N}  \int_\kappa^\infty \frac{\rd r}{r^{2+6\delta}} \leq C N \eps^2 \kappa^{-1-6\delta}\,. \] 
Hence,
\[ |\text{III}| \leq \eps \|  \cN^{1/2} \UN \xi_N \|^2 + C N \eps \kappa^{1-6\delta} + 
C N \eps^2 \kappa^{-1-6\delta}\,.  \]
Minimizing over $\kappa$ we find $\kappa = \eps^{1/2}$ and we conclude
\[ |\text{III}| \leq \eps \|  \cN^{1/2} \UN \xi_N \|^2 + C N \eps^{3/2-6\delta}\,. \]

Combining this bound with (\ref{eq:If}) and (\ref{eq:IIf}), we obtain from (\ref{eq:ddt}) that, for every $\delta > 0$ small enough, there is a constant $C > 0$ such that  
\[ \left| \frac{d}{dt} \langle \UN \xi_N, \cN \UN \xi_N \rangle \right| \leq 
C  \langle \UN \xi_N,  \cN \UN \xi_N \rangle + C N \eps^{1/2-\delta} \]
for all $t \in [0;T]$. Gronwall's lemma implies that there exists a constant $C > 0$ such that 
\[ \sup_{t\in [0;T]} \langle\UN \xi_N, \cN \UN \xi_N \rangle \leq C \left[ \langle \xi_N , \cN \xi_N \rangle + N \eps^{1/2-\delta}  \right]\,. \]
This concludes the proof of Proposition \ref{prop:Gronwall}. We still have to show Lemma \ref{lm:tr1}.

\begin{proof}[Proof of Lemma \ref{lm:tr1}] 
The integral kernel of the commutator $[\chi_{(r,z)}, \omega_{N,t}]$ is 
\[ \begin{split}
[ \chi_{\left(r,z\right)} , \omega_{N,t}] (x;y) &= \left( \chi_{\left(r,z\right)}(x) -\chi_{\left(r,z\right)}(y) \right) \omega_{N,t}(x;y)  \\
&= \int_0^1 \rd s\; \frac{\rd}{\rd s} e^{-\frac{(x-z)^2}{r^2}s}\omega_{N,t} (x;y)e^{-\frac{(x-z)^2}{r^2}(1-s)}  \\
&=  - \int_0^1 \rd s \;e^{-\frac{(x-z)^2}{r^2}s} \left[ \frac{(x-z)^2}{r^2}, \omega_{N,t} \right](x;y) \, e^{-\frac{(y-z)^2}{r^2}(1-s)}\,. 
\end{split} \]
Hence
\begin{equation}\label{eq:comm1}
\begin{split}
[ &\chi_{\left(r,z\right)} , \omega_{N,t}] 
\\ &= - \int_0^1 \rd s\; \chi_{\left(r/\sqrt{s},z\right)} (x) \left[ \frac{(x-z)^2}{r^2}, \omega_{N,t} \right] \chi_{\left(r/\sqrt{1-s},z\right)} (x) \\  
&= - \sum_{i=1}^3 \int_0^1 \rd s\; \chi_{\left(r/\sqrt{s},z\right)} (x) \\ & \hspace{1cm} \times \left( \frac{(x-z)_i}{r^2}\left[ x_i , \omega_{N,t} \right] + \left[ x_i , \omega_{N,t} \right] \frac{(x-z)_i}{r^2} \right) \chi_{\left(r/\sqrt{1-s},z\right)}(x) \\
&= \sum_{i=1}^3 \mathrm{I}_i + \mathrm{II}_i 
\end{split}
\end{equation}
where, with an abuse of notation, we use $\chi_{(.,.)} (x)$ to denote both the function of $x$ and the corresponding multiplication operator. 

We focus on the first term on the r.h.s. of (\ref{eq:comm1}), for example fixing $i=1$. The other components of the first term, and the three components of the second term can then be treated similarly. We use the spectral decomposition of the commutator $\left[ x_1 ,\omega_{N,t} \right]$ (which, by assumption, is trace class for all $t \in [0;T]$), given by 
\[ [x_1 ,\omega_{N,t}] = i \sum_j \lambda_j | \ph_j \rangle \langle \ph_j| \]
for a sequence of eigenvalues $\lambda_j \in \bR$ and an orthonormal system $\ph_j$ in $L^2 (\bR^3)$ (we introduced $i=\sqrt{-1}$ on the r.h.s., because the commutator is anti self-adjoint). We find 
\[ \begin{split} 
\text{I}_1 &= \int_0^1 \rd s\; \chi_{\left(r/\sqrt{s},z\right)}(x) \frac{(x-z)_1}{r^2} \left[ x_1, \omega_{N,t} \right]\chi_{\left(r/\sqrt{1-s},z\right)}(x)\\
&= \frac{i}{r} \sum_j \lambda_j \int_0^1 \frac{\rd s}{\sqrt{s}} \; \bigg\vert \, \chi_{\left(r/\sqrt{s},z\right)} (x) \frac{(x-z)_1}{r/\sqrt{s}} \varphi_j \bigg\rangle \bigg\langle \chi_{\left(r/\sqrt{1-s},z \right)} (x) \varphi_j \bigg\vert
\end{split} \]
and therefore, since $\| |\ph \rangle \langle \psi| \|_\text{tr} = \| \ph \| \| \psi \|$, 
\begin{equation}\label{eq:tr1-1} \begin{split} 
\| \text{I}_1 \|_\text{tr} &\leq \frac{1}{r} \sum_j |\lambda_j| \int_0^1 \frac{ds}{\sqrt{s}} \left\| \chi_{(r/\sqrt{s},z)} (x) \frac{|x-z|}{r/\sqrt{s}} \ph_j \right\| \, \left\|  \chi_{(r/\sqrt{1-s},z)} (x) \ph_j \right\|  \\ 
&\leq \frac{1}{r} \int_0^1 \frac{ds}{\sqrt{s}} \left( \sum_j |\lambda_j| \left\| \chi_{(r/\sqrt{s},z)} (x) \frac{|x-z|}{r/\sqrt{s}} \ph_j \right\|^2 \right)^{1/2} \\ &\hspace{4cm}\times \left( \sum_j |\lambda_j| \left\| \chi_{(r/\sqrt{1-s},z)} (x) \ph_j \right\|^2 \right)^{1/2}\,.\end{split} \end{equation}
We compute
\begin{equation}\label{eq:max1} \begin{split} 
\sum_j |\lambda_j| \left\| \chi_{(r/\sqrt{1-s},z)} (x) \ph_j \right\|^2  &= \int dx \, e^{-2(1-s)(x-z)^2/r^2} \rho_{|[x,\omega_{N,t}]|} (x)\\ &\leq C \frac{r^3}{(1-s)^{3/2}} \,  \rho^*_{|[x,\omega_{N,t}]|} (z) \end{split} \end{equation} 
where $\rho^*_{|[x_i,\omega_{N,t}]|}$ is the Hardy-Littlewood maximal function associated with $\rho_{|[x_i,\omega_{N,t}]|}$. To prove (\ref{eq:max1}), we write
\[ \begin{split} e^{-2(1-s)(x-z)^2/r^2} &= \int_0^1 \chi (t \leq e^{-2(1-s)(x-z)^2/r^2}) dt \\ &= \int_0^1 \chi \left( |x-z| \leq \sqrt{\frac{r^2 \log (1/t)}{2(1-s)}} \right) dt  \end{split} \]
and, using Fubini, we find 
\[\begin{split} \int dx \, &e^{-2(1-s)(x-z)^2/r^2} \rho_{|[x,\omega_{N,t}]|} (x) \\ &= \int_0^1 dt \int dx \, \chi \left( |x-z| \leq \sqrt{\frac{r^2 \log (1/t)}{2(1-s)}} \right) \,
\rho_{|[x,\omega_{N,t}]|} (x) \\ &\leq C \frac{r^3}{(1-s)^{3/2}} \,  \rho^*_{|[x,\omega_{N,t}]|} (z)  \int_0^1  (\log (1/t))^{3/2} \\ &\leq C \frac{r^3}{(1-s)^{3/2}} \,  \rho^*_{|[x,\omega_{N,t}]|} (z) \end{split} \]
which shows (\ref{eq:max1}). Similarly to (\ref{eq:max1}), we also find 
\[ \sum_j |\lambda_j| \left\| \chi_{(r/\sqrt{s},z)} (x) \frac{|x-z|}{r/\sqrt{s}} \ph_j \right\|^2 \leq C \frac{r^3}{s^{3/2}} \rho^*_{|[x_1,\omega_{N,t}]|} (z)\,. \]
Combining this bound with the simpler estimate
\[ \sum_j |\lambda_j| \left\| \chi_{(r/\sqrt{s},z)} (x) \frac{|x-z|}{r/\sqrt{s}} \ph_j \right\|^2 \leq C\sum_j |\lambda_j|  = \| \rho_{|[x_1,\omega_{N,t}]|} \|_1 \]
we obtain 
\[ \sum_j |\lambda_j| \left\| \chi_{(r/\sqrt{s},z)} (x) \frac{|x-z|}{r/\sqrt{s}} \ph_j \right\|^2 \leq
C \frac{r^{3\alpha} \, \| \rho_{|[x_1, \omega_{N,t}]|} \|^{1-\alpha}_1 }{s^{3\alpha/2}} \, \left( \rho^*_{|[x_1, \omega_{N,t}]|} (z) \right)^{\alpha} \]
for any $0 \leq \alpha \leq 1$. Inserting the last bound and (\ref{eq:max1}) on the r.h.s. of (\ref{eq:tr1-1}) we conclude  
\[ \begin{split} \| \text{I}_1 \|_\text{tr} &\leq C r^{(1 + 3\alpha)/2} \| \rho_{|[x_1, \omega_{N,t}]|} \|_1^{(1-\alpha)/2} \left( \rho^*_{|[x_1, \omega_{N,t}]|} (z) \right)^{(1+\alpha)/2} \\ &\hspace{4cm} \times \int_0^1 ds\frac{1}{s^{1/2+3\alpha/4} (1-s)^{3/4}}\,. \end{split} \]
Hence, for all $\delta > 0$ we find (putting $\alpha = 2/3 -2\delta$)
\[ \| \text{I}_1 \|_\text{tr} \leq  C r^{3/2 - 3\delta} \| \rho_{|[x_1, \omega_{N,t}]|} \|_1^{1/6+\delta} \left( \rho^*_{|[x_1, \omega_{N,t}]|} (z) \right)^{5/6 -\delta}
\]
which concludes the proof of Eq. (\ref{eq:tr1}), and of Lemma \ref{lm:tr1}.
\end{proof}

{\it Acknowledgements.} The work of M.P. has been carried out thanks to the financial support of the NCCR SwissMAP. C.S  gratefully acknowledges support by the Forschungskredit UZHFK-15-108. B.S. is happy to acknowledge support from the Swiss National Science Foundation through the SNF Grant ``Effective equations from quantum dynamics''.

\end{document}